\documentclass[a4paper]{article}

\usepackage[english]{babel}
\usepackage[ansinew]{inputenc}
\usepackage{amsmath}
\usepackage{amsthm}
\usepackage{amssymb}
\usepackage{amsfonts}
\usepackage{color}
\usepackage{graphicx}
\usepackage{authblk}

\usepackage[lofdepth,lotdepth]{subfig}
\usepackage{tikz}

\newcounter{saetning} 
  \newtheorem{theo}[saetning]{Theorem}

\theoremstyle{definition} 
	
\theoremstyle{remark}
	
	\newtheorem{example}{Example}

\title{Mutual information matrices are not always positive semi-definite}
\begin{document}
\date{}
\author{Sune K. Jakobsen\thanks{School of Mathematical Sciences and School of Electronic Engineering \&
Computer Science, Queen Mary University of London, Mile End Road, London,
E1 4NS, UK. Email: S.K.Jakobsen@qmul.ac.uk.}}

  \maketitle 
  
    \section*{Abstract}
    
    For discrete random variables $X_1,\dots, X_n$ we construct an $n$ by $n$ matrix. In the $(i,j)$-entry we put the mutual information $I(X_i;X_j)$ between $X_i$ and $X_j$. In particular, in the $(i,i)$-entry we put the entropy $H(X_i)=I(X_i;X_i)$ of $X_i$. This matrix, called the mutual information matrix of $(X_1,\dots,X_n)$, has been conjectured to be positive semi-definite. In this note, we give counterexamples to the conjecture, and show that the conjecture holds for up to three random variables.

    \section*{Introduction}
    
    For a  random variable\footnote{In this note, all random variables will assumed to be discrete.} $X$ taking the values $x_1,\dots x_k$, Shannon defined the entropy of $X$ to be 
    \[H(X)=-\sum_{i=1}^k p_i \log(p_i).\]
     Here $p_i$ denotes the probability $\Pr(X=x_i)$ and, throughout this note $\log$ denotes the base $2$ logarithm. For random variables $X_1,\dots ,X_n$ with some joint distribution, we define the entropy of two random variables $X_i,X_j$ by
     \[H(X_i,X_j):= H(X_{i,j})=H((X_i,X_j)).\]
That is, we consider the random variable $(X_i,X_j)$ that is the tuple of $X_i$ and $X_j$ and take the entropy of that random variable. Similarly for larger set of random variables. The mutual information of two random variables is defined by
\[I(X_i;X_j)=H(X_i)+H(X_j)-H(X_i,X_j).\]   
In particular, $I(X_i;X_i)=H(X_i)+H(X_i)-H(X_i;X_i)=H(X_i)$.

    For a tuple of random variables $(X_1,\dots, X_n)$ we define its mutual information (MI) matrix to be the $n$ by $n$ matrix whose $(i,j)$ entry is given by $I(X_i;X_j)$. This matrix was claimed to be positive semi-definite in \cite{CL}, but this has never been proved. In this note, we will show some counterexamples, and a proof that the conjecture is true for all three-tuples.

    \section*{Examples and a theorem}

    \begin{example}

    Let $X_1$ and $X_2$ be independent random variables each uniformly distributed on $\{0,1\}$. Let $X_3=X_1 + X_2 \text{(mod} 2\text{)}$ and let $X_4=(X_1,X_2)$. Now the mutual information matrix for $(X_1,X_2,X_3,X_4)$ is
    \[M=\left(\begin{matrix}
    1& 0 & 0 & 1\\
    0&1&0&1\\
    0&0&1&1\\
    1&1&1&2    
    \end{matrix}\right).\]
    This matrix has $\left(1,1,1,\frac{1-\sqrt{13}}{2}\right)^T$ as an eigenvector, and the corresponding eigenvalue is $\frac{3-\sqrt{13}}{2}< -.30277<0$ so it is not positive semi-definite.
    \end{example}
   Such an example can easily be extended to examples with more random variables. For example we could take the $X_1,X_2,X_3,X_4$ from above together with random variables $Y_1,Y_2,\dots Y_{n-4}$ independent from each other and from the $X_i$'s. 
   
   Here is an example with a negative eigenvalue that has a much larger absolute value.
   
   \begin{example}
   Let $X_1,\dots ,X_n$ be independent random variables each uniformly distributed on $\{0,1\}$. For any non-empty set $S\subseteq\{1,\dots,n\}$ define $Y_S=\sum_{i\in S}X_i$ mod $2$. It is clear these variables are pairwise independent. Consider the set of all the $Y_S$'s together with the random variable, given by the tuple $(X_1,\dots,X_n)$. This is a set of $2^n$ random variables and the MI matrix is
    \[M=\left(\begin{matrix}
    1& 0 & 0&\dots  & 1\\
    0&1&0&\dots &1\\
    0&0&1&\dots &1\\
    \vdots &\vdots &\vdots &\ddots & \vdots \\
    1&1&1&\vdots &n    
    \end{matrix}\right).\]
    We see that $(1,\,  \dots ,\, 1,\, x_n)^T$ with $x_n=\frac{n-1-\sqrt{(n-1)^2+4(2^n-1)}}{2}$ is an eigenvector and corresponding eigenvalue is $\frac{n+1-\sqrt{(n-1)^2+4(2^n-1)}}{2}$. The absolute value of the eigenvalues grows as $\Theta( 2^{n/2})$.
   \end{example}
   
   \begin{example}
   We know that for random variables $Y_1, \dots ,Y_n$ all constants $c\geq 0$ and all $\epsilon>0$ we can find random variables $Z_1,\dots,Z_n$ such that for each set $S\subseteq \{1,\dots n\}$ we have $|cH(Y_S)-H(Z_S)|<\epsilon$ \cite{ZY}. By using this in the above example, with $c=\frac{1}{n+1}$ we see that we can find random variables $Z_1,\dots Z_{2^n}$, such that $H(Z_1,\dots, Z_{2^n})\leq 1$ but minus the lowest eigenvalue of the mutual information matrix is $\Theta(\frac{ 2^{n/2}}{n})$. That is, even if we keep the total entropy of the random variables bounded, we can get arbitrarily low eigenvalues.
   \end{example}
   
   \begin{example}
  In some applications, there may be monotonicity constraints that would rule out the above examples:
  e.g. in the first example $X_3=X_1 + X_2 \text{(mod} 2\text{)}$, so if $X_1=0$ then $X_3$ is increasing in $X_2$ and if $X_1=1$, $X_3$ is decreasing in $X_2$. However, if we instead take $X_3=X_1+X_2$ the mutual information matrix becomes
       \[M=\left(\begin{matrix}
    1& 0 & 0.5 & 1\\
    0&1&0.5&1\\
    0.5&0.5&1.5&1.5\\
    1&1&1.5&2    
    \end{matrix}\right).\]
   
    which has an eigenvalue $\approx -0.11062$ and so is still not positive semi-definite. 
    \end{example}
    
    The following theorem shows that the conjecture is true for collections of up to $3$ random variables.
   
    \begin{theo}
    Let $M$ be the mutual information matrix of three random variables $X_1,X_2,X_3$. Then $M$ is positive semi-definite. 
    \end{theo}
    \begin{proof}
    Let $i,j,k$ be a permutation of $1,2,3$. We know that $I(X_i;X_j)\geq 0$, so $I(X_i;X_j|X_k)\geq I(X_i;X_j;X_k)$. So the information diagram (see \cite{RY} for an introduction to information diagrams) for $X_1,X_2,X_3$ can be written as in Figure \ref{fig:split} where $a$ and the $b_i$ are all non-negative. This can be achieved by setting $b_7=0$ if $I(X_1;X_2;X_3)\leq 0$ and $a=0$ if $I(X_1;X_2;X_3)\geq 0$.
    
    \begin{figure}
\centering
\begin{tikzpicture}[scale=0.8]
\draw (3,0) circle (3);
\draw (6,0) circle (3);
\draw (4.5,-2.5) circle (3);
\node at (1.5,0.3) {$b_1$};
\node at (4.5,1.1) {$a+b_2$};
\node at (7.5,0.3) {$b_3$};
\node at (4.5,-4) {$b_4$};
\node at (2.7,-2.1) {$a+b_5$};
\node at (6.3,-2.1) {$a+b_6$};
\node at (4.5,-1) {$-a+b_7$};
\node at (0,2.5) {$X_1$};
\node at (9,2.5) {$X_2$};
\node at (7,-5) {$X_3$};
\end{tikzpicture}
\caption{}
\label{fig:split}
\end{figure}

 The corresponding mutual information matrix is now
 \[M=a\left(\begin{matrix} 1&0&0\\0&1&0\\0&0&1\end{matrix}\right)+ b_1\left(\begin{matrix} 1&0&0\\0&0&0\\0&0&0\end{matrix}\right)+b_2\left(\begin{matrix} 1&1&0\\1&1&0\\0&0&0\end{matrix}\right)+\dots+b_7 \left(\begin{matrix} 1&1&1\\1&1&1\\1&1&1\end{matrix}\right).\]
    Thus $M$ can be written as a sum of positive semi-definite matrices, so it is itself positive semi-definite.

    \end{proof}
    
    \subsection*{Open problem}
    
    Empirically, D. Polani (personal communication) has observed that the mutual information matrix is positive semi-definite in many applications. It would be interesting to give a natural general sufficient condition that explains this phenomenon.

	 \end{document}